\documentclass[letterpaper, 10 pt, conference]{ieeeconf}  

\IEEEoverridecommandlockouts                              

\usepackage{latexsym,amsfonts}
\usepackage{paralist}
\usepackage{dsfont}
\usepackage{booktabs}
\usepackage{bm}
\usepackage{graphicx}
\usepackage{subfigure}
\usepackage{eso-pic}
\usepackage{mathtools}
\usepackage{epstopdf}
\usepackage{tikz}
\DeclareSymbolFont{symbolsC}{U}{pxsyc}{m}{n}
\SetSymbolFont{symbolsC}{bold}{U}{pxsyc}{bx}{n}
\DeclareFontSubstitution{U}{pxsyc}{m}{n}
\DeclareMathSymbol{\medcirc}{\mathbin}{symbolsC}{7}
\usetikzlibrary{arrows,calc,decorations.markings,math,arrows.meta}

\usepackage{algorithm}
\usepackage[noend]{algpseudocode}
\algnewcommand\algorithmicinput{\textbf{Input:}}
\algnewcommand\Input{\item[\algorithmicinput]}
\algnewcommand\algorithmicoutput{\textbf{Output:}}
\algnewcommand\Output{\item[\algorithmicoutput]}
\algnewcommand{\Initialize}[1]{%
	\State \textbf{Initialize:}
	\Statex \hspace*{\algorithmicindent}\parbox[t]{.8\linewidth}{\raggedright #1}
}

\newtheorem{theorem}{Theorem}[section]
\newtheorem{lemma}[theorem]{Lemma}
\newtheorem{problem}[theorem]{Problem}
\newtheorem{proposition}[theorem]{Proposition}
\newtheorem{remark}[theorem]{Remark}
\newtheorem{assumption}[theorem]{Assumption}

\newcommand{\R}{{\mathbb{R}}}

\newcommand{\e}{\mathsf{e}}

\newcommand{\ra}{\rightarrow}

\newcommand{\ol}{\overline}

\title{\LARGE \bf
Funnel-based Reachability Control of Unknown Nonlinear Systems\\ using Gaussian Processes
}

\author{Sandeep Gorantla$^{1,\dagger}$, Jeel Chatrola$^{1,\dagger}$, Jay Bhagiya$^{1,\dagger}$, Adnane Saoud$^{2}$, and Pushpak Jagtap$^{1}$
\thanks{*This work was supported in part by the Google Research Grant, the CSR Grant by Nokia Corporation, the SERB Start-up Research Grant, and by ANR PIA funding: ANR-20-IDEES-0002.}
\thanks{$^\dagger$ The authors contributed equally.}
\thanks{$^{1}$S. Gorantla, J. Chatrola, J. Bhagiya, and P. Jagtap are with the Robert Bosch Center for Cyber-Physical Systems, Indian Institute of Science, Bangalore, India.{\{sgorantla,pushpak,jeelchatrola,jaybhagiya\}@iisc.ac.in}}%
\thanks{$^2$Adnane Saoud is with Laboratoire des Signaux et Syst\`emes, CentraleSup\'elec, Universit\'e Paris Saclay, Gif-sur-Yvette, France.
{\tt\small adnane.saoud@centralesupelec.fr}}
}

\begin{document}

\maketitle
\thispagestyle{empty}
\pagestyle{empty}

\begin{abstract}
This paper aims to synthesize a reachability controller for an unknown dynamical system.  We first learn the unknown system using Gaussian processes and the (probabilistic) guarantee on the learned model.  Then we use the funnel-based controller synthesis approach using this approximated dynamical system to design the controller for a reachability specification.  Finally, the merits of the proposed method are shown using a numerical example.
\end{abstract}

\section{INTRODUCTION}

{Existing controller synthesis approaches generally rely on a mathematical model of the system, such as physics-based first principal models. The formal guarantees provided by the synthesized controller are valid as long as the considered dynamical model is accurate. When dealing with complex dynamical systems, describing the system in a closed-form model is often complicated. In this case, a common practice is to resort to data-driven techniques.}

{The Gaussian process is a non-parametric learning-based approach that provides probabilistic approach to approximate and to synthesize controllers for unknown systems \cite{kocijan2016modelling}.} {There are several works that utilize} GPs for providing MPC scheme \cite{hewing2019cautious}, adaptive control \cite{chowdhary2014bayesian}, tracking control\cite{beckers2019stable}, backstepping control  \cite{capone2019backstepping}, feedback linearization  \cite{umlauft2017feedback}, safe optimization of controller \cite{berkenkamp2016safe}, reinforcement learning \cite{akametalu2014reachability}, and control barrier functions for safety specification \cite{barriergppushpak}.

This work will consider the controller synthesis problem for reachability specification for unknown dynamical systems. In the past few decades, there have been several works in the literature addressing reachability problem (see \cite{ravanbakhsh2019formal,rungger2016scots,lygeros1999controllers,junges2021enforcing,vignali2016method}) for known dynamical systems. To solve this problem, we leverage the funnel-based control approaches \cite{bechlioulis2014low} that have been extensively used for controlling systems with prescribed performance constraints (see \cite{bu2021prescribed} and references therein for examples). 
We first employ the Gaussian process learning to approximate the system dynamics using the noisy measurements along a probabilistic bound on approximation. Then, the synthesis of a closed-form funnel-based control law that ensures the satisfaction of the reachability specification, with a given confidence, using the learned dynamics from the GP model, is presented. Finally, we show the validness of our approach using a numerical example.

{The organization of the paper is as follows. Section~\ref{problem} introduces some notations and present the main problem addressed in the paper. Section~\ref{gpm} explains how Gaussian processes make it possible to learn unknown dynamical systems along with statistical guarantees. Section~\ref{Sec3} presents a solution to the reachability problem by combining learned Gaussian processes with funnel-based control techniques. Finally, Section~\ref{case study} presents numerical results validating the merits of the proposed approach.}

\section{PROBLEM FORMULATION}\label{problem}
\subsection{Notations}

The set of real, positive real, nonnegative real, and positive integer numbers are represented using $\mathbb{R}$, $\mathbb{R}^+$, $\mathbb{R}_0^+ $, and $\mathbb{N}$, respectively. $\mathbb{R}^p$ denotes $p$-dimensional Euclidean space and $\mathbb{R}^{p \times q}$ denotes a space of real matrices with $p$ rows and $q$ columns. 
A diagonal matrix in $\R^{p\times p}$ with diagonal entries $d_1,\ldots, d_p$ is denoted by $diag\{d_1,\ldots, d_p\}$.
Given a matrix $M\in\R^{p\times q}$, $M^T$ represents transpose of matrix $M$.
For a vector $x=[x_1,\ldots,x_n]^T\in\R^n$, we denote $\mathsf{sign}(x)=[\mathsf{sign}(x_1),\ldots,\mathsf{sign}(x_n)]^T$, where $\mathsf{sign}(x_i)=\left\{\begin{matrix}
-1 & \text{if } x_i<0 \\ 
1 & \text{if } x_i\geq0 
\end{matrix}\right.$, we use $\|x\|$ and $\|x\|_\infty$ to denote Euclidean norm and infinity norm of a vector, respectively. 
We denote the empty set by $\emptyset$. We use $\bm I_p$ to represent the identity matrix in $\R^{p\times p}$.  
For $a,b\in\R$ and $a< b$, we use $(a,b)$ and $[a,b]$ to represent open and close intervals in $\R$, respectively. {For a function $f: \mathbb{R}^p \rightarrow \mathbb{R}^p$, $f_i:\mathbb{R}^p \rightarrow \mathbb{R}$, $i\in \{1,2,\ldots,p\}$ denotes the $i$-th component of $f$.}
Consider a set $X_a\subset\R^p$, its projection on $i$th dimension, where $i\in\{1,\ldots,p\}$, is given by an interval $[\underline X_{ai},\overline X_{ai}]\subset \R$, where $\underline X_{ai}:=\min\{x_i\in\R\mid[x_1,x_2,\ldots,x_p]\in X_a\}$, $\overline X_{ai}:=\max\{x_i\in\R\mid[x_1,x_2,\ldots,x_p]\in X_a\}$, and $\text{Int}(X_a)$ denotes interior of set $X_a$. $\mathcal{N}(m,C)$ denotes multivariate Gaussian distribution, where $m \in\R^p$ and $C\in\R^{p\times p}$ are mean and covariance matrices of appropriate sizes, respectively.
For events $B_1,\ldots, B_p$, $\bigcap_{i=1}^p B_i$ represents inner product of events.\\
\subsection{Problem Formulation}
Consider nonlinear control-affine system $\mathcal S$:
           	\begin{align}\label{sys}
           	    \dot{x} = f(x) + g(x)u,
           	\end{align}
   where $x\in X \subset \R^{n}$ is a state vector and $u$ is a control input. In this work, we assume that the map $f:X\rightarrow \mathbb{R}^n$ is unknown, the map $g:X\rightarrow \mathbb{R}^{n\times m}$ is known, and $g(x)g^{T}(x)$ is positive definite for all $x \in X$.	\\
Assumption imposing a restriction on the complexity of the map $f$ through reproducing kernel Hilbert space (RKHS) norm is described below.
\begin{assumption}\label{A2}
	For map $f:X\rightarrow \mathbb{R}^n$ in $\mathcal S$, $\|f_i\|_k\leq\infty$ for all $i\in\{1,\ldots,n\}$ (i.e., the RKHS norm w.r.t. kernel $k$ is bounded).
\end{assumption}  
Note that all continuous functions defined over compact state-space satisfy the above assumption for most of the commonly used kernels \cite{seeger2008information}. For more details on RKHS norm, we refer interested reader to \cite{paulsen2016introduction}. 
\begin{assumption}\label{A3}
	\cite{umlauft2018uncertainty} We have access to measurements $x\in X$ and $y=f(x)+w$, where $w \sim \mathcal{N}(0_n,\rho_f^2\bm I_n)$ is an additive noise with $\rho_f \in \mathbb{R}^+_0$.
\end{assumption}

 Next we formally define the controller synthesis problem for reachability specification.
 \begin{problem}\label{prob1}
	Given the system $\mathcal{S}$ with Assumptions \ref{A2}-\ref{A3}, sets $X_a,X_b\subseteq X$ goal, 
	{design a closed-form controller that provides a 
	lower bound on the probability of the trajectory $x_{x_0u}(t)$ for any $x_0\in X_a$ to reach $X_b$.}
\end{problem}
We use a funnel-based controller synthesis approach \cite{bechlioulis2014low} for designing the controller for the above problem using the learned dynamics through Gaussian processes.

\section{GAUSSIAN PROCESS APPROXIMATION}\label{gpm}
Gaussian processes (GPs) \cite{williams2006gaussian} is a non-parametric learning approach to approximate an unknown nonlinear function $f:X\ra\R^n$ using samples. 
The data we obtain for each component of $f$ can be viewed as a collection of random variables having a joint multivariate Gaussian
$\tilde f_i(x)\sim\mathcal{GP}(\mathsf{m}_i,K_i)$, $i\in\{1,\ldots,n\}$
where $m_i$ is the mean function which is set to 0 in practice, $K_{i}$ gives the covariance between $f_i(x)$ and $f_i(x')$ and is a function of corresponding $x,x' : K_{i}$ = $k(x,x')$ known as kernel function. The kernel function can be any kind, provided that it  generates a positive definite covariance matrix $K_i$ and is chosen according to problem. Some frequently used kernels include linear, squared exponential and Mat\`ern kernels  \cite{williams2006gaussian}. Complete approximation of $f$ with $n$ independent GPs is therefore given by,
\begin{align*}
\tilde f(x)=\left\{\begin{matrix}
\tilde f_1(x)\sim\mathcal{GP}(0,k_1(x,x')),\\ 
\vdots \\ 
\tilde f_n(x)\sim\mathcal{GP}(0,k_n(x,x')).
\end{matrix}\right.
\end{align*}

Now the posterior distribution of $f_i(x)$, conditioned on a given set of $N$ measurements $\{x^{(1)},\ldots,x^{(N)} \}$ and $\{y^{(1)},\ldots,y^{(N)}\}$, with $y^{(j)}=f(x^{(j)})+w^{(j)}$, $j\in\{1,\ldots,N\}$, is Gaussian with mean and covariance
\begin{align}
\mu_i(x)&=\overline k_i^T( K_i+\sigma_f^2\bm I_N)^{-1}y_i,\label{mean}\\
\sigma_i^2(x)&=k_i(x,x)-\overline k_i^T( K_i+\sigma_f^2\bm I_N)^{-1}\overline k_i,\label{SD}
\end{align}
where $\overline k_i=[k_i(x^{(1)},x),\cdots, k_i(x^{(N)},x)]^T\in\R^N$, $y_i=[y_i^{(1)}, \cdots, y_i^{(N)}]^T\in\R^N$, and 
\begin{align*}
K_i=\begin{bmatrix}
k_i(x^{(1)},x^{(1)})  &\cdots   & k_i(x^{(1)},x^{(N)}) \\ 
\vdots & \ddots  & \vdots \\ 
k_i(x^{(N)},x^{(1)}) &  \cdots & k_i(x^{(N)},x^{(N)}) 
\end{bmatrix}\in\R^{N\times N}. 
\end{align*}
We consider $\overline{\sigma}_i^2=\max_{x\in X}\sigma_i^2(x)$. 
One can readily see that such a bound exists as the set $X$ is compact and the kernels are continuous.
The approximation of overall $f$ is as follows:    
\begin{align}
\mu(x)&:=[\mu_1(x), \ldots, \mu_n(x)]^T,\label{eq:mu_gp}\\
\sigma^2(x)&:=[\sigma^2_1(x), \ldots, \sigma^2_n(x)]^T.\label{eq:rho_gp}
\end{align}
In following proposition, we provide a probabilistic bound on the difference between the inferred mean $\mu_i(x)$ and the true value of $f_i(x)$.
\begin{proposition}\label{lemma1}
	 Consider a system $\mathcal{S}$ with Assumptions \ref{A2} and \ref{A3}, and GP approximation with mean $\mu$ in \eqref{eq:mu_gp} and variance $\sigma^2$ in \eqref{eq:rho_gp} obtained using $N$ measurements. Then, the approximation error is bounded by
	\begin{align}\label{aaa}
    \mathbb{P}\Big\{\hspace{-0.1em}\mu(x)\hspace{-0.2em}-\hspace{-0.2em}\beta\sigma(x)\leq f(x)\leq\mu(x)\hspace{-0.2em}+\hspace{-0.2em}\beta\sigma(x), \hspace{-0.1em}\forall x\hspace{-0.2em}\in\hspace{-0.2em} X\hspace{-0.1em}\Big\}\hspace{-0.2em}\geq\hspace{-0.2em}(1\hspace{-0.2em}-\hspace{-0.2em}\varepsilon)^n,
    \end{align}
	with $\varepsilon\in(0,1)$ and $\beta=diag\{\beta_1,\ldots,\beta_n\}$, where $\beta_i:=\sqrt{2\|f_i\|_{k_i}^2+300\gamma_i \log^3(\frac{N+1}{\varepsilon})}$, where $\gamma_i$ denotes information gain $($c.f. Remark \ref{info_gain}$)$.
\end{proposition} 
\begin{proof}
	The proof can be found in \cite{barriergppushpak}. 
\end{proof}
    \begin{remark}\label{info_gain}
	The information gains $\gamma_i$ represents the maximum mutual information between data samples and unknown map $f_i$. Obtaining $\gamma_i$ is hard. However, there are techniques to over-approximate the value, see \cite{srinivas2012information} for example.  
\end{remark}
Moreover, computing bound on RKHS norm $\|f_i\|_{k_i}\leq B_i$ is also in general hard. However, considering Lipschitz-like assumption, one can compute $B_i$ as discussed below.
\begin{lemma}\cite[Lemma 1]{adnanesymbolic}\label{lemma_4}
Consider a kernel function $k_i$ and we assume that $f_i(x)$ satisfies $|f_i(x)-f_i(y)|\leq L_i\sqrt{\|x-y\|_\infty}$ for all $x,y\in X$, then $B_i=\frac{L_i}{\sqrt{2\|\frac{\partial k_i}{\partial x}\|_\infty}}$. 
\end{lemma}
Now by utilizing the upper bound on RKHS norm $B_i$ in Lemma \ref{lemma_4}, one can provide a deterministic bound (i.e. with probability 1) for unknown dynamics $f_i(\cdot)$ as discussed in the following result. 
\begin{lemma}\label{lemma_5}
Consider a system $\mathcal S$ with Assumption \ref{A2} and \ref{A3} and GP approximation with mean $\mu$ and standard deviation $\sigma$ as given in \eqref{eq:mu_gp} and \eqref{eq:rho_gp}, respectively. 
Then $\forall x\in X$, it follows that
\begin{align}\label{bounnds_1}
     \mu_i(x)+\tilde\beta_i\sigma_i(x) \leq f_i(x)\leq \mu_i(x)+\tilde\beta_i\sigma_i(x),
\end{align}
with $\tilde\beta_i = \sqrt{ B_i^{2}-y_i^T(K_i+\sigma_f^2\bm I_N)^{-1}y_i + N}$, where $B_i$ is an upper bound on RKHS norm $\|f_i\|_{k_i}$ as defined in Lemma \ref{lemma_4}; $y_i$ and $K_i$ are defined in \eqref{mean} and \eqref{SD}, respectively and $N$ is a number of data samples.
\end{lemma}
\begin{proof}
The proof is similar to that of \cite[Lemma 2]{adnanesymbolic} and is omitted here.
\end{proof}
Note that the bound obtained in \eqref{bounnds_1} are very conservative. For more discussion, please refer to the case study in Section \ref{case study}.

\section{REACHABILITY USING FUNNEL-BASED CONTROL}\label{Sec3}


In this section, we propose the use of funnel-based control approach \cite{bechlioulis2014low} to solve Problem \ref{prob1}. Consider a funnel representing time-varying bounds for the trajectory $x_i, i\in\{1,\ldots,n\}$ given as follows
\begin{align}\label{funnel}
-c_i\rho_i(t)<x_i(t)-\eta_i<d_i\rho_i(t)
\end{align}
for all $t\in\R_0^+$, where $\rho_i:\R_0^+\ra\R^+$, $i\in\{1,\ldots,n\}$ are positive, smooth, and strictly decreasing funnel functions, $c_i,d_i\in\R_0^+$ and $\eta_i\in\R$ are some constants. In this work, we consider the following form of funnel function 
\begin{align}
\label{funnel2}
\rho_i(t)=\rho_{i0}\e^{-\epsilon_i t}+\rho_{i\infty},
\end{align} 
where $\rho_{i0}$, $\rho_{i\infty}$, $\epsilon_i\in\R^+$ are positive constants and $\rho_{i\infty}=\lim_{t\ra\infty}\rho_i(t)$. Now, by normalizing $ x_i(t)-\eta_i$ with respect to the performance function $\rho_i(t)$, the modulating error is defined as $\hat x_i(t):=\frac{x_i(t)-\eta_i}{\rho_i(t)}$ and the corresponding performance region $\hat{\mathcal{D}}_i:=\{\hat x_i\mid\hat x_i\in(-c_i,d_i)\}$. Then, we transform the modulated error through a strictly increasing transformation function $T_i:\hat{\mathcal{D}}_i\ra\R$ such that $T_i(0)=0$ and is chosen as
\begin{align}\label{transformation_fun}
T_i(\hat x_i)=\ln \Big(\frac{d_i(c_i+\hat x_i)}{c_i(d_i-\hat x_i)}\Big).
\end{align}
The transformed error is then defined as $\xi_i(x_i(t),\rho_i(t)):=T_i(\hat x_i)$. It can be verified that if the transformed error is bounded, then the  modulated error $\hat x_i$ is constrained within the region $\hat{\mathcal{D}}_i$. This also implies that $x_i(t)-\eta_i$ evolves within the bounds given in \eqref{funnel}. Differentiating $\xi_i$ with respect to time, we obtain transformed error dynamics for $i$th dimension as
\begin{align}
\dot{\xi}_i=\phi_i(\hat x_i,t)[\dot{x}_i+\alpha_i(t)(x_i-\eta_i)],
\end{align}
where $\phi_i(\hat x_i,t):=\frac{1}{\rho_i(t)}\frac{\partial T_i(\hat x_i)}{\partial \hat{x}_i}>0$ for all $\hat x_i\in(-c_i,d_i)$ and $\alpha_i(t):=-\frac{\dot{\rho}_i(t)}{\rho_i(t)}>0$ for all $t\in\R_0^+$ are the normalized Jacobian of the transformation function $T_i$ and the normalized derivative of the performance function $\rho_i$, respectively. Now, by stacking all the transformed error dynamics, one gets  
\begin{align}
\dot{\xi}=\Phi_t(\dot{x}+\alpha_t(x-\eta)),
\end{align}
where $\xi=[\xi_1,\ldots,\xi_n]^T$, $\Phi_t=diag\{\phi_1(\hat{x}_1,t),\ldots,\phi_n(\hat{x}_n,t)\}$, $\alpha_t=diag\{\alpha_1(t),\ldots$, $\alpha_n(t)\}$, and $\eta=[\eta_1,\ldots, \eta_n]^T$.
The following theorem provides the result for enforcing reachability specification by utilizing the funnel approach. 

\begin{theorem}\label{thm1}
	Consider the system $\mathcal{S}$, the learned GP approximation with mean $\mu$ \eqref{eq:mu_gp} and standard deviation $\sigma$ \eqref{eq:rho_gp}, sets $X_a,X_b\subset X$, $\Xi_i:=[\underline X_{ai},\overline X_{ai}]\cap[\underline X_{bi},\overline X_{bi}]$, $\underline{X}_i:=\min\{\underline X_{ai},\underline X_{bi}\}$, $\overline X_{i}:=\max\{\overline X_{ai},\overline X_{bi}\}$, an arbitrarily chosen state $\eta=[\eta_1,\eta_2,\ldots,\eta_n]^T\in \text{Int}(X_b)$ satisfying 
	\begin{align*}
	\eta_i\in\left\{\begin{matrix}
	\Xi_i & \text{if } \Xi_i\neq\emptyset\\ 
	[\underline X_{bi},\overline X_{bi}] & \text{if } \Xi_i=\emptyset,
	\end{matrix}\right.
	\end{align*}
	$ i\in\{1,\ldots,n\}$, and funnel function {(\ref{funnel2})} with $\epsilon_i\in\R^+$,
	\begin{align*}
	\rho_{i0}=\left\{\begin{matrix}
	\max\{|\eta_i-\underline X_{ai}|,|\eta_i-\overline X_{ai}|\} & \text{if } \Xi_i \neq\emptyset\\ 
	\max\{|\eta_i-\underline{X}_i|,|\eta_i-\overline X_{i}|\} & \text{if } \Xi_i=\emptyset,
	\end{matrix}\right.
	\end{align*}
	constants $c_i, d_i$ as follows:
	\begin{align*}
	&c_i=\frac{|\eta_i-\underline X_{ai}|}{\rho_{i0}}, d_i=\frac{|\eta_i-\overline X_{ai}|}{\rho_{i0}}, & \text{if} \quad \Xi_i\neq\emptyset ; \\
	&c_i=\frac{|\eta_i-\underline{X}_i|}{\rho_{i0}}, d_i=\frac{|\eta_i-\overline X_{i}|}{\rho_{i0}},
	& \text{if} \quad \Xi_i=\emptyset;
	\end{align*}
	and $\rho_{i\infty}$ is such that $\prod\limits_{i\in\{1,\ldots,n\}}\eta_i+[-c_i\rho_{i\infty},d_i\rho_{i\infty}] \subset X_b$\footnote{One can choose $\rho_{i\infty}$ arbitrary small in order to satisfy this condition}. \\
	Then under time-varying control law:
	\begin{align}\label{contr}
    u(x,\rho) =& -g(x)^T(g(x)g(x)^T)^{-1}(\mu(x) \nonumber\\&+ (\mathsf{sign}(x-\eta))^{T}\beta\sigma(x)+\xi(x,\rho) 
	+\ol\epsilon (x-\eta)),	
	\end{align}
	where\\ $\xi(x,\hspace{-.1em}\rho)\hspace{-.2em}=\hspace{-.2em}[\xi_1(x_1,\hspace{-.1em}\rho_1\hspace{-.1em}),\ldots,\xi_n(x_n,\hspace{-.2em}\rho_n\hspace{-.1em})]^T$
	$\hspace{-.3em}:=\hspace{-.3em}\Bigg[\hspace{-.2em}\ln\hspace{-.1em} \Big(\hspace{-.1em}\frac{d_1\big(\hspace{-.1em}c_1+\frac{x_1-\eta_1}{\rho_1}\hspace{-.1em}\big)}{c_1\big(\hspace{-.1em}d_1-\frac{x_1-\eta_1}{\rho_1}\hspace{-.1em}\big)}\hspace{-.1em}\Big),\ldots, \ln \hspace{-.1em} \Big(\hspace{-.1em}\frac{d_n\big(\hspace{-.1em}c_n+\frac{x_n-\eta_n}{\rho_n}\hspace{-.1em}\big)}{c_n\big(\hspace{-.1em}d_n-\frac{x_n-\eta_n}{\rho_n}\hspace{-.1em}\big)}\hspace{-.1em}\Big)\hspace{-.2em}\Bigg]^T$ is a transformation error as discussed above, $\ol\epsilon:=\max_{i\in\{1,\ldots,n\}}\epsilon_i$, $\mathsf{sign}(x-\eta) = [\mathsf{sign}(x_1-\eta_1),\dots,\mathsf{sign}(x_n-\eta_n)]^T,\beta=diag\{\beta_1,\dots,\beta_n\}$, one can ensure that $\exists t\in\R_0^+$ such that $x_{x_0u}(t)\cap X_b\neq\emptyset$ for all $x_0\in X_a$ with probability $(1 - \epsilon)^n$. In other words, the trajectory starting from any initial point in $X_a$, will reach $X_b$ in a finite time under the control law \eqref{contr} with a minimum probability of $(1-\epsilon)^{n}$. 
\end{theorem}

\begin{proof}
{To improve readability, we will drop the arguments $x$ and $\rho$ of the map $\xi$.}
	Consider Lyapunov like function $V=\frac{1}{2}\xi^T\xi$ and
	\begin{align}
	\dot{V}=&\xi^T\Phi_t(f(x)+g(x)u+\alpha_t(x-\eta))\nonumber\\
	=&\xi^T\Phi_t(f(x)-g(x)g(x)^T(g(x)g(x)^T)^{-1}(\mu(x) \nonumber\\
	&+(\mathsf{sign}(x-\eta))^{T}\beta\sigma(x)+\xi+\ol\epsilon (x-\eta))+\alpha_t(x-\eta))\nonumber\\
	=&-\xi^T\Phi_t(\mu(x) + (\mathsf{sign}(x-\eta))^{T}\beta\sigma(x)-f(x))\nonumber\\  & -\xi^T\Phi_t\xi- \ol\epsilon\xi^T\Phi_t(x-\eta)  +\xi^T\Phi_t\alpha_t(x-\eta).\label{xyz}
	\end{align}
	Considering the construction of transformed error $\xi$ and \eqref{aaa}, one can obtain that the first term of last equality is always non-positive with a probability greater than $(1-\epsilon)^n$. To elaborate more, for an $i\in\{1,\ldots,n\}$, we consider the following two cases:\\
	\textit{Case I:} $\xi_i<0$ implies that $(x_i-\eta_i)<0$ (this is due to $\xi_i(\hat x_i)$ is strictly increasing and $\xi_i(0)=0$). It follows that
	\begin{align*}
	    -\xi_i\phi_i&(\hat{x}_i,t)(\mu_i(x)+\mathsf{sign}(x_i-\eta_i)\beta_i\sigma_i(x)-f{_i}(x))\\&=-\xi_i\phi_i(\hat{x}_i,t)(\mu_i(x)-\beta_i\sigma_i(x)-f{_i}(x))\leq0.
	\end{align*}  
	The last inequality is due to $\xi_i<0$, $\phi_i(\hat x_i,t)>0$, and $\mu_i(x)-\beta_i\sigma_i(x)-f_i(x)\leq0$.\\
	\textit{Case II:} $\xi_i\geq0$ implies that $(x_i-\eta_i)\geq 0$. It follows that
	\begin{align*}
	    -\xi_i\phi_i&(\hat{x}_i,t)(\mu_i(x)+\mathsf{sign}(x_i-\eta_i)\beta_i\sigma_i(x)-f_i(x))\\&=-\xi_i\phi_i(\hat{x}_i,t)(\mu_i(x)+\beta_i\sigma_i(x)-f_i(x))\leq0.
	\end{align*} 
	The last inequality is due to $\xi_i>0$, $\phi_i(\hat x_i,t)>0$, and $\mu_i(x)+\beta_i\sigma_i(x)-f_i(x)\geq0$.
	This implies that the first term of 
	\eqref{xyz} is non-positive with probability of at least $(1-\epsilon)^n$.
	
	Next, following the facts that $\Phi_t$ and $\alpha_t$ are positive definite matrices, $\alpha_t<\ol\epsilon:=\max_{i\in\{1,\ldots,n\}}\epsilon_i$, $\xi^T(x-\eta)\geq0$ (this is due to $\xi_i(\hat x_i)$ is strictly increasing and $\xi_i(0)=0$), one obtains $\dot{V}\leq-\xi^T\Phi_t\xi$. This implies that $\xi(t)$ is bounded for all $t\in\R_0^+$ and hence we guarantee \eqref{funnel} that is $-c_i\rho_i(t)+\eta_i<x_i(t)<d_i\rho_i(t)+\eta_i $ with probability of at least $(1-\epsilon)^n$. From the choice of $\eta$ and constants $\rho_{i0}$, $\rho_{i\infty}$, $c_i$, $d_i$, $\eta_i$ for all $i\in\{1,\ldots,n\}$, one can readily ensure that $X_a\subseteq\prod\limits_{i\in\{1,\ldots,n\}}[-c_i\rho_i(0)+\eta_i,-d_i\rho_i(0)+\eta_i]$ and as $\lim\limits_{t\ra\infty}\prod\limits_{i\in\{1,\ldots,n\}}[-c_i\rho_i(t)+\eta_i,-d_i\rho_i(t)+\eta_i]=\prod\limits_{i\in\{1,\ldots,n\}}\eta_i+[-c_i\rho_{i\infty},d_i\rho_{i\infty}] \subset X_b$. This implies that there exist $ t\in\R_0^+$ such that $x_{x_0u}(t)\cap X_b\neq\emptyset$ for all $x_0\in X_a$ with probability of at least $(1-\epsilon)^n$. This concludes the proof.
\end{proof}
\section{CASE STUDY}\label{case study}
Here, we demonstrates the efficacy of the proposed result using a numerical example adapted from \cite{umlauft2018uncertainty}.
\begin{figure*}[t] 
 	\centering
 	\subfigure[The original $f(x)$]
 	{\includegraphics[scale=0.55]{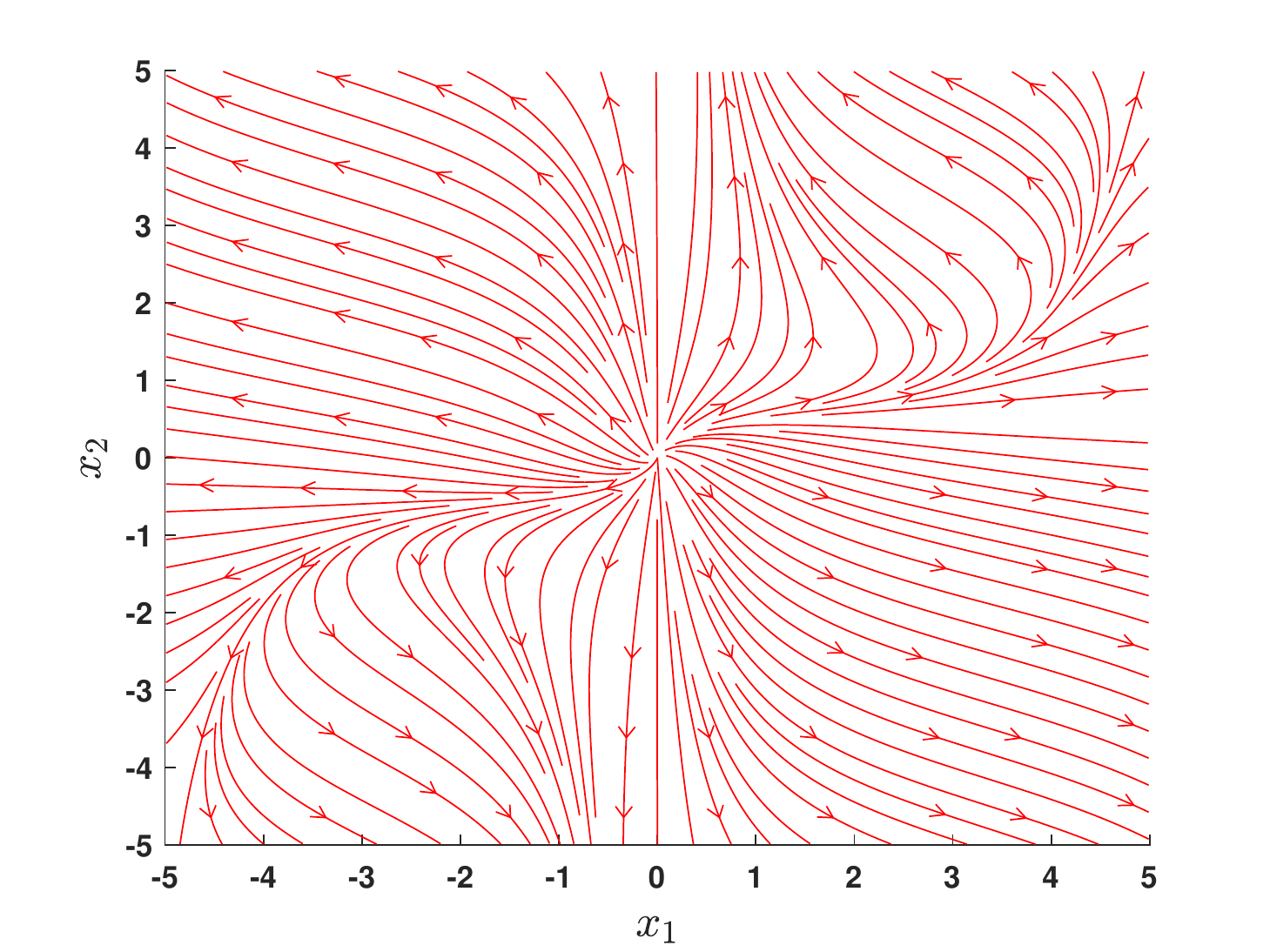}}
 	\subfigure[GP approximation of $f(x)$]
 	{\includegraphics[scale=0.55]{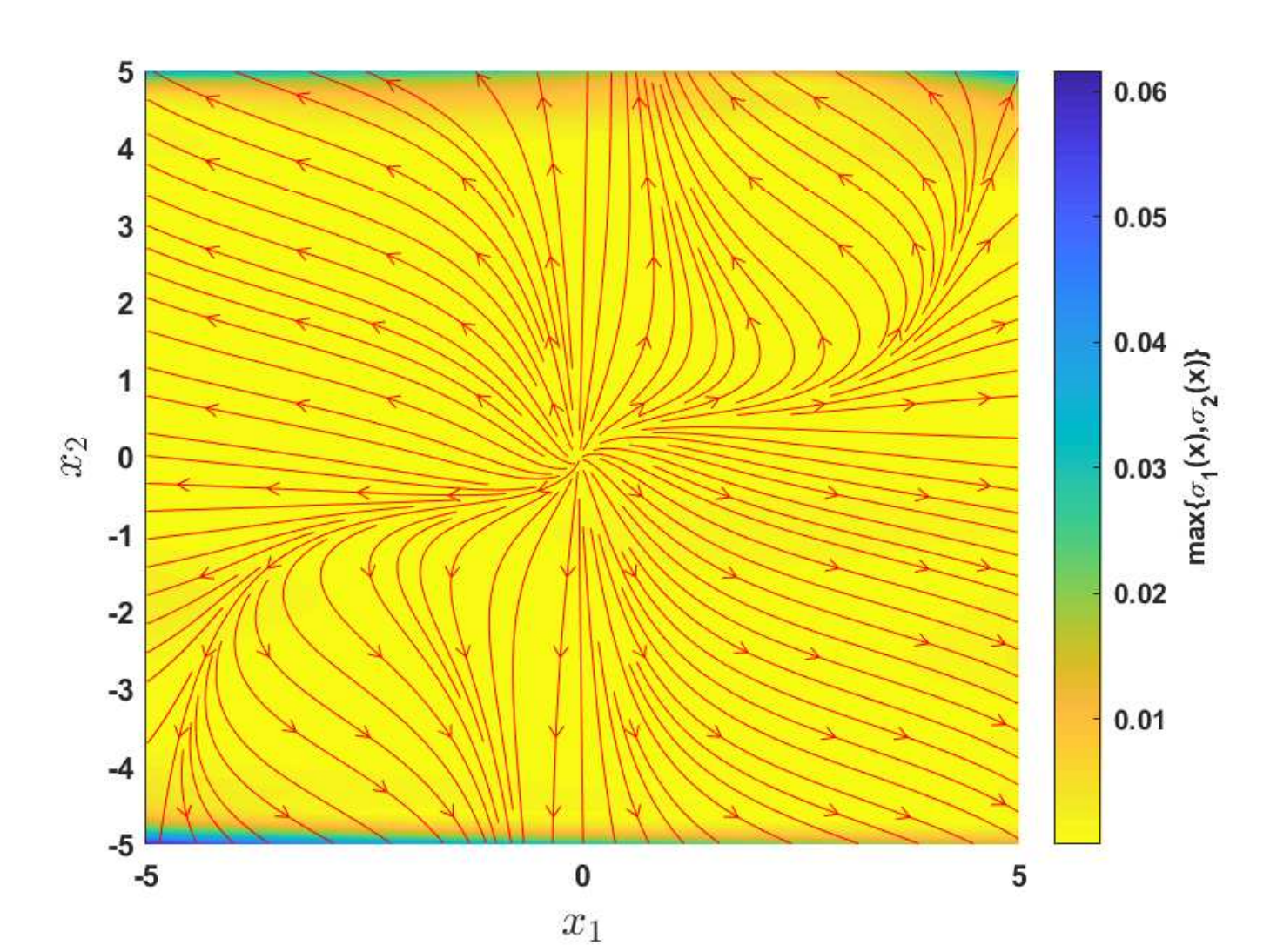}} 
 	\vspace{1em}
 	\caption{The GP approximation for considered example. the color-map shows maximum of standard deviations.}
 	\label{fig:learned_model}
\end{figure*}
\begin{align*}
f(x)=\begin{bmatrix}
f_1(x)\\ 
f_2(x)
\end{bmatrix}=\begin{bmatrix}
x_1 + (\cos(x_1) - 1)x_2 \\ 
-s(x_1) + x_2
\end{bmatrix}\hspace{-0.2em}, g(x)=\begin{bmatrix}
1 & 0 \\
0 & 1
\end{bmatrix}\hspace{-0.2em},
\end{align*}
where $s(x_1) = \frac{1}{1+exp(-2x_1)} - 0.5$ is the shifted sigmoid function. We consider a compact state-space $X=[-5,5]\times[-5,5]$, initial state-set $X_a=[-2,-3]\times[-2,-3]$, and the goal set $X_b=[1,3]\times[1,3].$ The functions $f_1(\cdot)$ and $f_2(\cdot)$ are continuous, has therefore a finite RKHS norm under the squared exponential kernel on a compact set and complies to Assumption \ref{A2}.  

For solving Problem \ref{prob1}, we first approximated the unknown dynamics using GPs 
with $50$ collected measurments of $x$ and corresponding $y$'s, $y=f(x)+w$, where $w \sim \mathcal{N}(0,\sigma_f^2\bm I_2)$, $\sigma_f=0.01$, by running the simulated system with different start states. We used exponential quadratic kernel  \cite{williams2006gaussian}
defined as 
$k_i(x,x')=\sigma_{k_i}^2\exp\Big(\sum_{j=1}^2\frac{(x_i-x_j')^2}{-2l_{ij}^2}\Big), i\in\{1,2\}$, where $\sigma_{k_1}=316$ and $\sigma_{k_2}=25.3$ are signal variances and $l_{11}=2.9$, $l_{12}=177$, $l_{21}=1.67$, and $l_{22}=50.5$ are length scales. We use Limited Memory Broyden–Fletcher–Goldfarb–Shanno (L-BFGS-B) algorithm \cite{doi:10.1137/0916069,10.1145/279232.279236} to obtain these parameters. The inferred mean and variance are as in \eqref{eq:mu_gp} and \eqref{eq:rho_gp} with $\overline{\sigma}_{\max}=\max\{\overline{\sigma}_1,\overline{\sigma}_2\} =0.0616 , \overline{\sigma}_{1} = 0.022 , \overline{\sigma}_{2} = 0.0616.$ 
Figure \ref{fig:learned_model} depicts the original and the approximated map $f(x)$.


Computing $\|f_i\|_{k_i}$ and $\gamma_j$, $i\in\{1,2\}$, is intractable in general. Thus, we used Monte-Carlo method to get the probability bound for the confidence interval given in Proposition \ref{lemma1}. 

 For a fixed value of $\beta_i\overline\sigma_i=$ $0.04$, $i=1,2$, we get a probability interval for the probability in \eqref{aaa} as $\mathbb{P}\Big\{\{\mu(x) -\beta \sigma(x)\leq f(x)\leq\mu(x) +\beta \sigma(x)\}, \forall x\in X\Big\}\in[0.9894,0.9907]$ with confidence $1-10^{-10}$ using $10^6$ realizations.
Thus, one can choose the lower bound $(1-\epsilon)^2$ as $0.9894$.

We also computed the value of $\tilde\beta_1=7.0878$ and $\tilde\beta_2=7.0710$ as shown in Lemma \ref{lemma_5}. To compare the conservativeness of the bounds, we compare the value of $\tilde\beta_i\ol\sigma_i$ with Monte-Carlo approach to obtain probability of 1 with confidence of $1-10^{-10}$. Using Monte-Carlo approach, the obtained values are $\tilde\beta_1\ol\sigma_1=0.016$ and $\tilde\beta_2\ol\sigma_2=0.0442$ and the values obtained using results of Lemma \ref{lemma_5} are $\tilde\beta_1\ol\sigma_1=0.1559$ and $\tilde\beta_2\ol\sigma_2=0.4366$, respectively. One can readily see the conservatism in the bounds obtained using results of Lemma \ref{lemma_5}.

\begin{figure}[H] 
 	\centering
  	\includegraphics[scale=0.45]{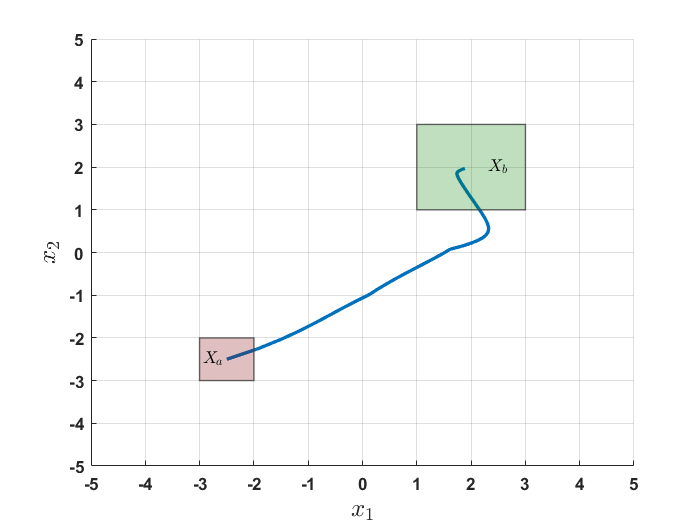}
    \caption{Simulation of the proposed funnel-based controller using the learned GPs, $X_a$ and $X_b$ are initial and goal sets, respectively. blue line indicates the state trajectory.}
 	\label{fig:Controller}
\end{figure}
With the help of learned mean and variance, we simulate the results using the proposed control law \eqref{contr}. The parameters of the controllers and construction of corresponding funnel functions are as per the Theorem \ref{thm1}.
\begin{figure*}[t] 
 	\centering
    \vspace{1em}
  	\subfigure[]
 	{\includegraphics[scale=0.55]{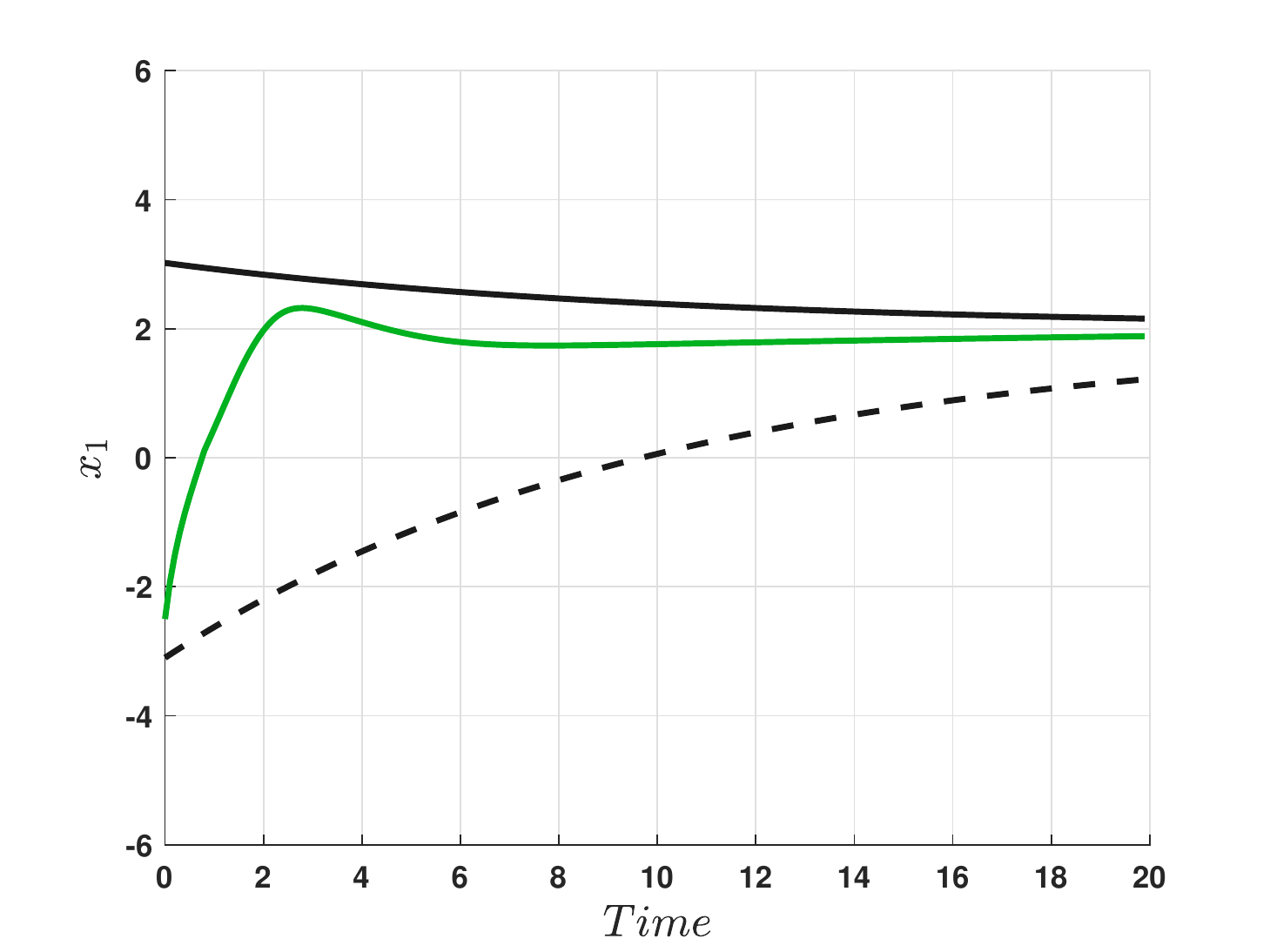}}
 	\vspace{1em}
 	\subfigure[]
 	{\includegraphics[scale=0.55]{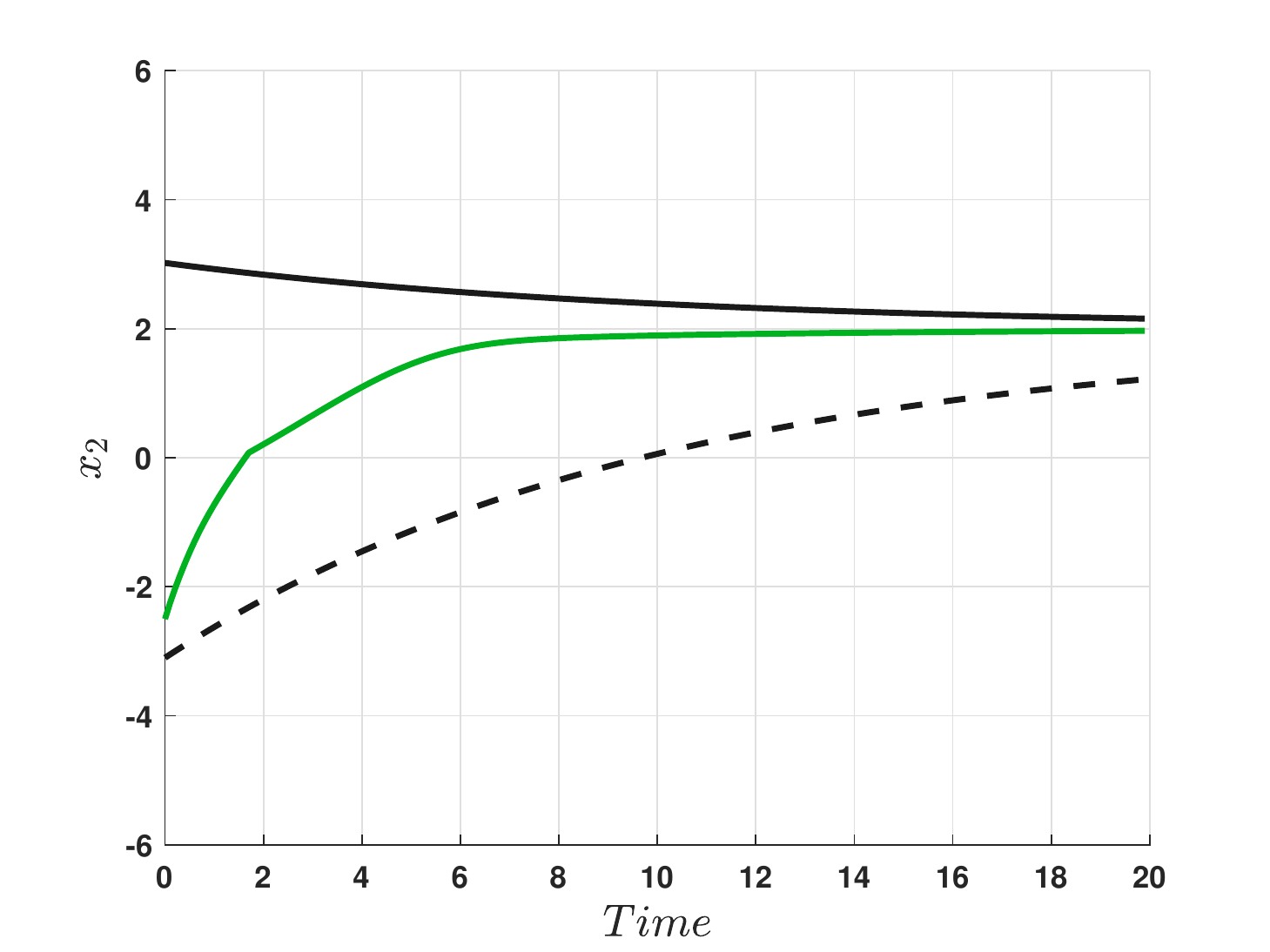}} 
 	\caption{Illustration of trajectories $x_1$ and $x_2$. The green solid lines represent trajectories, black solid and dashed lines represent upper and lower bounds of funnels, respectively.}
 	\label{fig:bounds}
\end{figure*}

The trajectory of the system reaching $X_b$ from $X_a$ is shown in Figure \ref{fig:Controller}. One can readily see from Figure \ref{fig:bounds}, the trajectories $x_1$ and $x_2$ satisfy the constructed funnel bounds.




\section{CONCLUSION AND FUTURE WORK}\label{conclusion}
The work proposed a scheme for designing closed-form controller for unknown nonlinear control systems enforcing reachability specifications. We provide a control policy using a funnel-based approach by approximating unknown system dynamics using Gaussian processes. We verified the proposed method using a numerical example. Future research includes incorporating constraints on the input space and extending the results to more complex specifications like linear/signal temporal logic specifications.


\bibliographystyle{IEEEtran}
\bibliography{bibliography}

\begin{thebibliography}{10}
\providecommand{\url}[1]{#1}
\csname url@samestyle\endcsname
\providecommand{\newblock}{\relax}
\providecommand{\bibinfo}[2]{#2}
\providecommand{\BIBentrySTDinterwordspacing}{\spaceskip=0pt\relax}
\providecommand{\BIBentryALTinterwordstretchfactor}{4}
\providecommand{\BIBentryALTinterwordspacing}{\spaceskip=\fontdimen2\font plus
\BIBentryALTinterwordstretchfactor\fontdimen3\font minus
  \fontdimen4\font\relax}
\providecommand{\BIBforeignlanguage}[2]{{%
\expandafter\ifx\csname l@#1\endcsname\relax
\typeout{** WARNING: IEEEtran.bst: No hyphenation pattern has been}%
\typeout{** loaded for the language `#1'. Using the pattern for}%
\typeout{** the default language instead.}%
\else
\language=\csname l@#1\endcsname
\fi
#2}}
\providecommand{\BIBdecl}{\relax}
\BIBdecl

\bibitem{kocijan2016modelling}
J.~Kocijan, \emph{Modelling and control of dynamic systems using {G}aussian
  process models}.\hskip 1em plus 0.5em minus 0.4em\relax Springer, 2016.

\bibitem{hewing2019cautious}
L.~Hewing, J.~Kabzan, and M.~N. Zeilinger, ``Cautious model predictive control
  using {G}aussian process regression,'' \emph{IEEE Transactions on Control
  Systems Technology}, 2019.

\bibitem{chowdhary2014bayesian}
G.~Chowdhary, H.~A. Kingravi, J.~P. How, and P.~A. Vela, ``Bayesian
  nonparametric adaptive control using {G}aussian processes,''
  \emph{transactions on neural networks and learning systems}, vol.~26, no.~3,
  pp. 537--550, 2014.

\bibitem{beckers2019stable}
T.~Beckers, D.~Kuli{\'c}, and S.~Hirche, ``Stable gaussian process based
  tracking control of euler--lagrange systems,'' \emph{Automatica}, vol. 103,
  pp. 390--397, 2019.

\bibitem{capone2019backstepping}
A.~Capone and S.~Hirche, ``Backstepping for partially unknown nonlinear systems
  using gaussian processes,'' \emph{Control Systems Letters}, vol.~3, no.~2,
  pp. 416--421, 2019.

\bibitem{umlauft2017feedback}
J.~Umlauft, T.~Beckers, M.~Kimmel, and S.~Hirche, ``Feedback linearization
  using {G}aussian processes,'' in \emph{56th Annual Conference on Decision and
  Control (CDC)}.\hskip 1em plus 0.5em minus 0.4em\relax IEEE, 2017, pp.
  5249--5255.

\bibitem{berkenkamp2016safe}
F.~Berkenkamp, A.~P. Schoellig, and A.~Krause, ``Safe controller optimization
  for quadrotors with {G}aussian processes,'' in \emph{International Conference
  on Robotics and Automation (ICRA)}.\hskip 1em plus 0.5em minus 0.4em\relax
  IEEE, 2016, pp. 491--496.

\bibitem{akametalu2014reachability}
A.~K. Akametalu, J.~F. Fisac, J.~H. Gillula, S.~Kaynama, M.~N. Zeilinger, and
  C.~J. Tomlin, ``Reachability-based safe learning with {G}aussian processes,''
  in \emph{53rd Conference on Decision and Control}.\hskip 1em plus 0.5em minus
  0.4em\relax IEEE, 2014, pp. 1424--1431.

\bibitem{barriergppushpak}
P.~Jagtap, G.~J. Pappas, and M.~Zamani, ``Control barrier functions for unknown
  nonlinear systems using {G}aussian processes,'' in \emph{59th Conference on
  Decision and Control (CDC)}.\hskip 1em plus 0.5em minus 0.4em\relax IEEE,
  2020, pp. 3699--3704.

\bibitem{ravanbakhsh2019formal}
H.~Ravanbakhsh, S.~Sankaranarayanan, and S.~A. Seshia, ``Formal policy learning
  from demonstrations for reachability properties,'' in \emph{International
  Conference on Robotics and Automation (ICRA)}.\hskip 1em plus 0.5em minus
  0.4em\relax IEEE, 2019, pp. 6037--6043.

\bibitem{rungger2016scots}
M.~Rungger and M.~Zamani, ``{SCOTS: A} tool for the synthesis of symbolic
  controllers,'' in \emph{Proceedings of the 19th international conference on
  hybrid systems: Computation and control}, 2016, pp. 99--104.

\bibitem{lygeros1999controllers}
J.~Lygeros, C.~Tomlin, and S.~Sastry, ``Controllers for reachability
  specifications for hybrid systems,'' \emph{Automatica}, vol.~35, no.~3, pp.
  349--370, 1999.

\bibitem{junges2021enforcing}
S.~Junges, N.~Jansen, and S.~A. Seshia, ``Enforcing almost-sure reachability in
  pomdps,'' in \emph{International Conference on Computer Aided
  Verification}.\hskip 1em plus 0.5em minus 0.4em\relax Springer, 2021, pp.
  602--625.

\bibitem{vignali2016method}
R.~Vignali and M.~Prandini, ``A method for detecting relevant inputs while
  satisfying a reachability specification for piecewise affine systems,'' in
  \emph{Conference on Control Applications (CCA)}.\hskip 1em plus 0.5em minus
  0.4em\relax IEEE, 2016, pp. 538--543.

\bibitem{bechlioulis2014low}
C.~P. Bechlioulis and G.~A. Rovithakis, ``A low-complexity global
  approximation-free control scheme with prescribed performance for unknown
  pure feedback systems,'' \emph{Automatica}, vol.~50, no.~4, pp. 1217--1226,
  2014.

\bibitem{bu2021prescribed}
X.~Bu, ``Prescribed performance control approaches, applications and
  challenges: {A} comprehensive survey,'' \emph{Asian Journal of Control},
  2021.

\bibitem{seeger2008information}
M.~W. Seeger, S.~M. Kakade, and D.~P. Foster, ``Information consistency of
  nonparametric {G}aussian process methods,'' \emph{Transactions on Information
  Theory}, vol.~54, no.~5, pp. 2376--2382, 2008.

\bibitem{paulsen2016introduction}
V.~I. Paulsen and M.~Raghupathi, \emph{An introduction to the theory of
  reproducing kernel Hilbert spaces}.\hskip 1em plus 0.5em minus 0.4em\relax
  Cambridge University Press, 2016, vol. 152.

\bibitem{umlauft2018uncertainty}
J.~Umlauft, L.~P{\"o}hler, and S.~Hirche, ``An uncertainty-based control
  {Lyapunov} approach for control-affine systems modeled by {G}aussian
  process,'' \emph{IEEE Control Systems Letters}, vol.~2, no.~3, pp. 483--488,
  2018.

\bibitem{williams2006gaussian}
C.~K. Williams and C.~E. Rasmussen, \emph{Gaussian processes for machine
  learning}.\hskip 1em plus 0.5em minus 0.4em\relax MIT press Cambridge, MA,
  2006, vol.~2, no.~3.

\bibitem{srinivas2012information}
N.~Srinivas, A.~Krause, S.~M. Kakade, and M.~W. Seeger, ``Information-theoretic
  regret bounds for {Gaussian} process optimization in the bandit setting,''
  \emph{Transactions on Information Theory}, vol.~58, no.~5, pp. 3250--3265,
  2012.

\bibitem{adnanesymbolic}
K.~Hashimoto, A.~Saoud, M.~Kishida, T.~Ushio, and D.~Dimarogonas,
  ``Learning-based symbolic abstractions for nonlinear control systems,'' in
  \emph{arXiv preprint arXiv:2004.01879}, 2021.

\bibitem{doi:10.1137/0916069}
\BIBentryALTinterwordspacing
R.~H. Byrd, P.~Lu, J.~Nocedal, and C.~Zhu, ``A limited memory algorithm for
  bound constrained optimization,'' \emph{SIAM Journal on Scientific
  Computing}, vol.~16, no.~5, pp. 1190--1208, 1995. [Online]. Available:
  \url{https://doi.org/10.1137/0916069}
\BIBentrySTDinterwordspacing

\bibitem{10.1145/279232.279236}
C.~Zhu, R.~H. Byrd, P.~Lu, and J.~Nocedal, ``Algorithm 778: {L-BFGS-B:} fortran
  subroutines for large-scale bound-constrained optimization,'' \emph{ACM
  Trans. Math. Softw.}, vol.~23, no.~4, p. 550–560, dec 1997.

\end{thebibliography}

\addtolength{\textheight}{-12cm}   

\end{document}